\documentclass{article}
\usepackage{amssymb,alltt,graphics,mathdots,amsmath,amsthm,fullpage,palatino}
\usepackage[usenames]{color}
\usepackage[table]{xcolor}

\usepackage[utf8]{inputenc}
\theoremstyle{theorem}
\newtheorem{lemma}{Lemma}
\newtheorem{theorem}[lemma]{Theorem}
\newtheorem{exam}[lemma]{Example}
\newenvironment{example}{\begin{exam}\rm}{\end{exam}}
\newtheorem{defi}[lemma]{Definition}
\newenvironment{definition}{\begin{defi}\rm}{\end{defi}}

\theoremstyle{definition}
\newtheorem*{remark}{Remark}

\newcommand\mut[1]{}
\newcommand\rank{\mbox{rank}}
\newcommand\alert[1]{{\it #1}}

\begin{document}

\title{A Decoding Approach to Reed--Solomon Codes \\from Their Definition}
\markright{A Decoding Approach to Reed-Solomon Codes}
\author{Maria Bras-Amor\'os\thanks{M. Bras-Amor\'os is with Universitat Rovira i Virgili, Tarragona, Catalonia  (e-mail: maria.bras@urv.cat)}}

\maketitle

\begin{abstract}
  Because of their importance in applications and their quite simple definition, Reed--Solomon codes can be explained in any introductory course on coding theory. However, decoding algorithms for Reed--Solomon codes are far from being simple and it is difficult to fit them in introductory courses for undergraduates. We introduce a new decoding approach, in a self-contained presentation, which we think may be appropriate for introducing error correction of Reed--Solomon codes to nonexperts. In particular, we interpret Reed--Solomon codes by means of the degree of the interpolation polynomial of the code words and from this derive a decoding algorithm. Compared to the classical algorithms, our algorithm appears to arise more naturally from definitions and to be easier to understand. It is related to the Peterson--Gorenstein--Zierler algorithm (see \cite{GorensteinZierler} and \cite{Peterson}).
  
\end{abstract}

\section{Introduction.}
Error control codes are used to detect and correct errors that may occur in data transmission or storage through eventually defective channels or storage devices that can distort the sent or stored information. For example, the atmosphere introduces errors in the transmission of images from the Meteosat satellite to Earth, different interferences in communications by mobile phone may cause transmission errors, or reading devices need correcting algorithms for handling CDs, DVDs, or USB memories. Error control codes are also used in distributed data storage in the cloud to recover lost or damaged chunks of information.
In the words of Elwyn R. Berlekamp, one of the fathers of coding theory, \begin{quote}Communication links transmit information from here to there. Computer memories transmit information from now to then. In either case, noise causes the received data to differ slightly from the original data. As Shannon \cite{Shannon1948} showed in 1948, the noise need not cause any degradation in reliability. The noise does impose some limiting capacity on the throughput rate, although that limit is typically well above the throughput rate at which real systems operate. Error-correcting codes enable a system to achieve a high degree of reliability despite the presence of noise \cite{Berlekamp1980}.\end{quote}

The {\it modus operandi} of those codes is to send along with the original information a small amount of redundancy, so that from all the received information one can deduce what is actually transmitted. The simplest example is adding for every transmitted bit (a $0$ or a $1$) two identical copies. If the original bit or one of its copies is received with an error, we can still correct it from the other two, which we expect to coincide. Note that by adding redundancy, on one side we improve the quality of the received information. But, on the other side, we augment the transmission cost. In the example of repeating bits, the transmission cost is multiplied by three.

Coding theory aims at designing and implementing codes with good correcting capacity, while maintaining a low transmission cost, as well as designing detection and correction algorithms that allow the receiver to recover the original information.

Berlekamp's reference \cite{Berlekamp1980} gives a detailed historical review (up to 1980) of coding theory since Shannon's cornerstone contribution \cite{Shannon1948}. At that time, the so-called Reed--Solomon codes \cite{ReedSolomon} and the most relevant algorithms for decoding Reed--Solomon codes had already appeared. They are the most universal error control codes and are currently being used directly or indirectly in most transmission devices and storage systems. Reed--Solomon codes admit different definitions as will be explained in this article, and they are all based in polynomials of bounded degrees over a finite field. One way to explain how they work is as follows. Fix a finite field of cardinality $q$. From the data one wants to transmit (say $k$ elements of ${\mathbb F}_q$), one interpolates a polynomial of degree less than $k$ that takes these values when evaluated at $k$ given nonzero elements of the finite field. Then one adds to the original $k$ information values the redundancy which consists of the evaluation of the polynomial at the remaining $q-1-k$ nonzero values of the finite field. Basic polynomial theory shows how any small part of the whole $(q-1)$-length of the received information can be restored from the rest.

Because of their importance in applications and their quite simple definition, Reed--Solomon codes can be explained in any introductory course. However, decoding algorithms for Reed--Solomon codes are far from being so simple and it is difficult to explain them in introductory courses for undergraduates. This is why we introduce our new decoding approach, in a self-contained presentation, which we think may be appropriate to introduce error correction of Reed--Solomon codes to nonexperts. Although a direct implementation of the algorithm presented in this article may not be as efficient as the most efficient known algorithms, we think that it is performable by any undergraduate student using basic software tools. However, we do not rule out the possibility that technical improvements to the algorithm may make it much more efficient, especially if one can deemphasize matrices in favor of polynomials.

The most celebrated algorithms for decoding Reed--Solomon codes have been the Peterson--Gorenstein--Zierler algorithm \cite{Peterson,GorensteinZierler} for its simplicity, and the algorithms designed to solve Berlekamp's key equation \cite{Berlekamp:Book}. The two primary decoding algorithms that solve Berlekamp's key equation are the Berlekamp-Massey algorithm \cite{Massey:Shift} and the Sugiyama et al. adaptation of the Euclidean algorithm \cite{Sugiyama:key}. The alternative so-called Welch--Berlekamp equations are solved in the Welch--Berlekamp algorithm \cite{WelchBerlekamp}. Bit-serialized multiplication and bit-serial encoders are more efficient for hardware implementation of shift registers \cite{Berlekamp1982}. This is used in the algorithm in \cite{Berlekamp1996}. The Welch--Berlekamp equations were also solved by Chambers' algorithm \cite{Chambers} and by Fedorenko's algorithm \cite{Fedorenko}.
Another general perspective is that of decoding pairs \cite{Pellikaan92,Pellikaan96}.
Guruswami and Sudan presented their breakout algorithm \cite{Sudan,GuruswamiSudan} decoding beyond half the minimum distance by means of list decoding.
All these algorithms and their relationships are analyzed in several papers such as \cite{Dornstetter,HeyJensen:BMandE,Meteer,BossertBezzateev,MateerRS}.

In Section~\ref{sec:rs} we revisit the definition of Reed--Solomon codes, giving four different, but equivalent, versions. Reinterpreting a definition related to the degree of the interpolation polynomial, we derive a decoding algorithm. The key result for the new formulation is Theorem~\ref{t:lambda} in Section~\ref{sec:decap}. Now, for correcting a received word, its interpolation polynomial is split into two parts, one with low order terms (lower than the code dimension) and the other one with the remaining terms. The latter part is fixed while the first part is modified in order to maximize the number of nonzero roots. This gives the code word at minimum distance from the received word.

Our decoding algorithm is related to the Peterson--Gorenstein--Zierler algorithm. We compare both algorithms in Sections \ref{sec:comp} and \ref{sec:opt} and see how our algorithm is well suited for the optimistic view of best case decoding \cite{Berlekamp1996}. This is the case when error correction codes of high correction capability are used, but with a low expectation of errors.

\section{Some background on coding theory.}
Let us start with some background definitions and results on coding theory.
Standard references are \cite{MacWilliamsSloane,Roman,JustesenHoholdt,Bierbrauer,Roth}.

\paragraph{The alphabet ${\mathbb F}_q$.}
The symbols that contain the information that needs to be sent as well as the symbols corresponding to the transformed and transmitted data are the elements of a finite field, which is often called the transmission {\it alphabet}. One can consider the case in which ${\mathbb F}_q$ is a prime field, that is, $q$ is a prime number and ${\mathbb F}_q$ can be identified by the set $\{0,1,\dots,q-1\}$, equipped with the usual addition and product modulo $q$. There will always exist an element $\alpha$ in ${\mathbb F}_q$ such that all the powers of $\alpha$ with exponent smaller than $q-1$ are different. Then, ${\mathbb F}_q=\{0,1,\alpha,\alpha^2,\dots,\alpha^{q-2}\}$. In this case, $\alpha$ is called a {\it primitive element}.

\begin{example}
\label{e:fset}
Consider ${\mathbb F}_7$. It is the set $\{0,1,2,3,4,5,6\}$ equipped with the addition and multiplication operations, always modulo $7$. For instance, in ${\mathbb F}_7$, $4+5=2$, $1-2=6$, $3\cdot 5=5$.
It is easy to verify that $\alpha=5$ is a primitive element of ${\mathbb F}_7$.
\end{example}

\paragraph{Linear codes.}
A  \alert{linear code} $C$ of \alert{length} $n$ over a finite field ${\mathbb F}_q$ is a vector subspace of ${\mathbb F}_q^n$. Its vectors are called \alert{code words}.
The \alert{dimension} $k$ of the code is the dimension of the subspace. In particular, the number of code words of $C$ is $q^k$. 

\paragraph{Generator matrices.}
We say that a matrix  $G$ of $k$ rows and $n$ columns is a  \alert{generator matrix} of $C$ if its rows are a set of vectors generating the code. The generator matrix is not unique, for example we can permute the rows. To encode a word of  $k$ symbols of  ${\mathbb F}_q$, we multiply it by the generator matrix.
\begin{example}
\label{e:Gf}
The following matrix is the generator matrix of a code $C$ of length $6$ and dimension $2$ over ${\mathbb F}_7$.
$$G_{ex}=\left(
  \begin{array}{cccccc}
    1 & 1 & 1 & 1 & 1 & 1 \\
    1 & 5 & 4 & 6 & 2 & 3 \\
    \end{array}
  \right).
  $$
  To encode the information $110256$  we split it into blocks of $k=2$ symbols and multiply each block by $G_{ex}$.
$$\left(\begin{array}{cc}
    1& 1
  \end{array}\right)
  G_{ex}
  =\left(\begin{array}{ccccccc}
    2& 6& 5& 0& 3& 4\\ 
    \end{array}\right),$$
$$\left(\begin{array}{cc}
  0& 2
\end{array}\right)
G_{ex}
    =\left(\begin{array}{ccccccc}
 2& 3& 1& 5& 4& 6\\
    \end{array}\right),$$
$$\left(\begin{array}{cc}
  5&6
    \end{array}\right) G_{ex}
      =\left(\begin{array}{ccccccc}
4& 0& 1& 6& 3& 2\\
    \end{array}\right).$$
      The encoded information will then be $265034231546401632$.  
\end{example}

\paragraph{Dual code and parity-check matrices.}
Consider the scalar product of two vectors $(u_0,u_1,\dots,u_{n-1})$ and $(v_0,v_1,\dots,v_{n-1})$ of ${\mathbb F}_q^n$, defined as $u_0 v_0+u_1 v_1+\dots+u_{n-1} v_{n-1}\in{\mathbb F}_q$. The \alert{dual code} of $C$ is  $C^\perp=\{v\in{\mathbb F}_q^n:v\cdot c=0 \mbox{ for all }c\in C\}$.
It is a linear code with the same length as  $C$ and dimension $n-k$. We can define it from a system of linear equations with coefficient matrix $G$. A matrix $H$ generating $C^\perp$ is called a \alert{parity-check-matrix} of  $C$.
Equivalently, a parity-check matrix of  $C$ is a matrix such that the code $C$ can be redefined as
$C=\{c\in{\mathbb F}_q^n:c\cdot h=0 \mbox{ for every row } h\mbox{ of }H\}$.
\begin{example}\label{e:Hf}
  The following matrix is a parity-check matrix of the code $C$ of Example~\ref{e:Gf}.    $$H_{ex}=\left(
  \begin{array}{cccccc}
    1 & 5 & 4 & 6 & 2 & 3 \\
    1 & 4 & 2 & 1 & 4 & 2 \\
    1 & 6 & 1 & 6 & 1 & 6 \\
    1 & 2 & 4 & 1 & 2 & 4 \\
  \end{array}
  \right)$$
\end{example}

\paragraph{Hamming distance, correction capability, and Singleton bound.}
The \alert{Hamming distance} between two words of the same length is the number of positions in which their symbols differ.
The purpose of decoding algorithms is, given an input vector $u$ of the same length as the code, output a code word $c\in C$ minimizing the Hamming distance between $u$ and $c$.
The \alert{weight} of a word is the number of nonzero symbols or, equivalently, its Hamming distance from the zero vector.
The \alert{minimum distance} $d$ of a linear code $C$ can be equivalently defined as (i) the minimum Hamming distance between two words of $C$; (ii) the minimum weight of nonzero words of $C$; (iii) the minimum number of linearly dependent columns of $H$.
The minimum distance of a code is an important parameter quantifying the error correction capability of the code. Indeed, if at most $\lfloor\frac{d-1}{2}\rfloor$ errors are added to a code word $c\in C$, corrupting it into a word $u$, then $c$ is the unique code word of $C$ at Hamming distance at most $\lfloor\frac{d-1}{2}\rfloor$ from $u$, and in this sense we say that $\lfloor\frac{d-1}{2}\rfloor$ errors can be corrected.

The \alert{Singleton bound} states that for a linear code of length $n$ and minimum distance $d$, the dimension $k$ satisfies $k\leq n-d+1.$ The codes attaining this bound are called \alert{maximum distance separable codes} (MDS).

\begin{example}\label{e:dmin}
The Hamming distance between the code words $265034$ and $231546$ of the code $C$ in Example~\ref{e:Gf} is $5$. The Hamming distance between the code word $111111$ corresponding to the first row of $G_{ex}$ and $265034$ is $6$. 

  Notice that the elements of the first row of the generator matrix of $C$ are all equal while the elements of the second row are all different.
  Any code word of $C$ will be a multiple of the second row plus a multiple of the first row. The components of any multiple of the second row, except for the zero vector, will all be different by field properties. Similarly, if we add to a vector whose components are all different a constant vector, then the components of the vector so obtained will also all be different by field properties.
  So, for any vector in $C$, either it is constant or all its components are different. This makes the Hamming distance between any two words in $C$ either equal to $6$ or to $6-1=5$. Consequently, the minimum distance of $C$ is~$5$.
\end{example}

\paragraph{Vandermonde matrices.}
Although Vandermonde matrices can be defined over any field, for our purposes we concentrate on finite fields.
Given $\alpha_1,\alpha_2,\dots,\alpha_n\in{\mathbb F}_q$, the Vandermonde matrix of $\alpha_1,\dots,\alpha_n$ of order $r$ is
defined as 
$$V_r(\alpha_1,\alpha_2,\dots,\alpha_n)=\left(
\begin{array}{cccc}
1 & 1 & \dots & 1\\
\alpha_1 & \alpha_2 & \dots & \alpha_n\\
\alpha_1^2 & \alpha_2^2 & \dots & \alpha_n^2\\
\vdots & \vdots & \ddots & \vdots\\
\alpha_1^{r-1} & \alpha_2^{r-1} & \dots & \alpha_n^{r-1}\\
\end{array}
\right).$$
\noindent It can be proved that the determinant of $V_n(\alpha_1,\alpha_2,\dots,\alpha_n)$ satisfies
$\left|V(\alpha_1,\alpha_2,\dots,\alpha_n)\right|=\prod_{1\leq i< j\leq n}(\alpha_j-\alpha_i).$
Consequently, $V_n(\alpha_1,\alpha_2,\dots,\alpha_n)$ has an inverse matrix if and only if $\alpha_i\neq\alpha_j$ for all $1\leq i< j\leq n$.

\section{Four definitions of Reed--Solomon codes.}
\label{sec:rs}

Let us introduce Reed--Solomon codes from four different, but complementary, points of view. 

\subsection{Reed--Solomon codes from generator matrices.}

Let ${\mathbb F}_q$ be the field with $q$ elements ($q$ a prime power)
and let $\alpha$ be a primitive element of ${\mathbb F}_q$. Then ${\mathbb F}_q=\{0,1,\alpha,\alpha^2,\dots,\alpha^{q-2}\}$. Let $n=q-1$.

\begin{definition}The Reed--Solomon code over ${\mathbb F}_q$ and of dimension $k$, $RS_{q,\alpha}(k)$, is the linear code of ${\mathbb F}_q^n$ with generator matrix 
\begin{equation*}G=\left(\begin{array}{ccccc}
1      &        1    & 1        & \dots &          1\\
1      & \alpha      & \alpha^2      & \dots & \alpha^{n-1}\\
1      & \alpha^2    & \alpha^4      & \dots & \alpha^{2(n-1)}\\
\vdots & \vdots      & \vdots        & \ddots& \vdots\\
1      & \alpha^{k-1} & \alpha^{(k-1)2} & \dots & \alpha^{(k-1)(n-1)}
\end{array}
\right).\end{equation*}
\end{definition}

\begin{example}
  \label{e:toyG}
  Consider the finite field ${\mathbb F}_7$. As noted above, the element $5\in{\mathbb F}_7$ is primitive. Indeed,   $5^0 = 1$,
    $5^1  =5$,
    $5^2  =4$,
    $5^3  =6$,
    $5^4  =2$,
    $5^5  =3$ and $5^6$ is again $1$.
  The code $RS_{7,5}(2)$ is exactly the code $C$ of Example~\ref{e:Gf}.
\end{example}

\subsection{Reed--Solomon codes from parity-check matrices.}

Consider the matrix

\begin{equation}
\label{e:HHHH}
  H=\left(\begin{array}{ccccc}
1      & \alpha      & \alpha^2      & \dots & \alpha^{n-1}\\
1      & \alpha^2    & \alpha^4      & \dots & \alpha^{2(n-1)}\\
\vdots & \vdots      & \vdots        & \ddots& \vdots\\
1      & \alpha^{n-k} & \alpha^{(n-k)2} & \dots & \alpha^{(n-k)(n-1)}
\end{array}
  \right).\end{equation}
\noindent It is a matrix of maximum rank (namely $n-k$) because of the Vandermonde structure.
Furthermore, the product of matrix $G$ and the transpose of matrix $H$ is the zero matrix. Indeed, 
the product of the $i$th row of matrix $G$ ($1\leq i\leq k$) times the $j$th row of matrix $H$ ($1\leq j\leq n-k$) is $\sum_{r=1}^{n}\alpha^{(i-1)(r-1)}\alpha^{j(r-1)}=\sum_{r=1}^{n}\alpha^{(i+j-1)(r-1)}.$ Now, because of the limits of $i$ and $j$, we have that $i+j-1<q-1$ and so $\alpha^{i+j-1}\neq 1$. Finally, the sum equals $\frac{(\alpha^{i+j-1})^n-1}{\alpha^{i+j-1}-1}=0$.

This enables us to give the following equivalent definition.

\begin{definition}\label{def:H}The Reed--Solomon code over ${\mathbb F}_q$ and of dimension $k$, $RS_{q,\alpha}(k)$, is the linear code of ${\mathbb F}_q^n$ 
with parity-check matrix equal to $H$.
\end{definition}

\begin{example}
  One can check that matrix $H_{ex}$ in Example~\ref{e:Hf} is of the form of the matrix $H$ in \eqref{e:HHHH}.
  \end{example}

Lemma \ref{l:mds} uses Definition \ref{def:H} to deduce that Reed--Solomon codes attain the Singleton bound, and so they are maximum distance separable (MDS) codes.
\begin{lemma}\label{l:mds}
The minimum distance of $RS_{q,\alpha}(k)$ is exactly $n-k+1$. Hence, it is an MDS code.
\end{lemma}

\begin{proof}
The submatrix given by any subset of $n-k$ columns (with column indices $0\leq j_1,\dots,j_{n-k}\leq {n-1}$) has determinant 
$$ \left|  
\begin{array}{cccc}
\alpha^{j_1} & \alpha^{j_2} & \dots & \alpha^{j_{n-k}}\\
\alpha^{2j_1} & \alpha^{2j_2} & \dots & \alpha^{2j_{n-k}}\\
\alpha^{3j_1} & \alpha^{3j_2} & \dots & \alpha^{3j_{n-k}}\\
\vdots & \vdots & \ddots & \vdots\\
\alpha^{(n-k)j_1} & \alpha^{(n-k)j_2} & \dots & \alpha^{(n-k)j_{n-k}}\\
\end{array}
\right|=\alpha^{j_1}\dots\alpha^{j_{n-k}}\cdot
\left|
V_n(\alpha^{j_1},\dots,\alpha^{j_{n-k}})\right|,
$$
which is not zero. So, any set of $n-k$ columns of the parity-check matrix are independent, and so the minimum distance must be at least $n-k+1$. By the Singleton bound the minimum distance must be exactly equal to $n-k+1$.
\end{proof}

\begin{example}
The minimum distance of $RS_{7,5}(2)$ is $5$ as justified in Example~\ref{e:dmin}. This equals $n-k+1=6-2+1$.
\end{example}

\subsection{Reed--Solomon codes and interpolation polynomials.}

Consider the set ${\mathbb F}_q[x]^{<k}$ of all polynomials with coefficients in ${\mathbb F}_q$ and of degree strictly less than $k$.
A general element $a\in{\mathbb F}_q[x]^{<k}$ is of the form $a=a_0+a_1x+\dots+a_{k-1}x^{k-1}$ with $a_i\in{\mathbb F_q}$. Observe that evaluating $a$ at $\alpha^{i-1}$ gives $a(\alpha^{i-1})=a_0+a_1\alpha^{i-1}+\dots+a_{k-1}\alpha^{(i-1)(k-1)},$ which is exactly the result of the product of the vector $(a_0,\dots,a_{k-1})$ by the $i$th column of matrix $G$. So, the product of the vector $(a_0,\dots,a_{k-1})$ by matrix $G$ is exactly the vector $(a(1),a(\alpha),a(\alpha^2),\dots,a(\alpha^{n-1})).$

\begin{definition}\label{def:interpolation}
  The Reed--Solomon code over ${\mathbb F}_q$ and of dimension $k$, $RS_{q,\alpha}(k)$, is the set 
$\{(a(1),a(\alpha),a(\alpha^2),\dots,a(\alpha^{n-1})): a\in{\mathbb F}_q[x]^{<k}\}.$
\end{definition}

\begin{example}\label{e:intpolGf}
  The three code words computed in Example~\ref{e:Gf}, which are  $265034$, $231546$, and $401632$ are, respectively, the evaluation of the polynomials $x+1$, $2x$, and $6x+5$ at $5^0,5^1,5^2,5^3,5^4,5^5$. The degree of the three polynomials is less than $2$ which is the dimension of the code.
\end{example}

Now, for each vector
$u=(u_0,\dots,u_{n-1})$ in ${\mathbb F}_q^n$, 
there exists a unique polynomial $f_u$ 
of degree at most $n-1$
such that $f_u(\alpha^i)=u_i$ for all $i$ in $\{0,\dots,n-1\}$.
It can be computed using the formula
$f_u=\sum_{i=0}^{n-1}u_i f_i$, where $f_i$ is the interpolation polynomial of the $i$th standard basis vector, that is, $f_i=\prod_{\substack{j=0\\j\neq i}}^{n-1}\frac{x-\alpha^j}{\alpha^i-\alpha^j}.$
The uniqueness of $f_u$ is a consequence of the fact that if $f_u=a_0+a_1x+\dots+a_{n-1}x^{n-1}$,
then the coefficients $a_0,\dots,a_{n-1}$ are a solution of 
the linear system of equations
$$
\left(\begin{array}{cccccc}
1      &    1        &    1          &    1           & \dots & 1      \\          
1      & \alpha      & \alpha^{2}     & \alpha^{3}     & \dots & \alpha^{(n-1)}\\
1      & \alpha^{2}  & \alpha^{4}      & \alpha^{6}     & \dots & \alpha^{2(n-1)}\\
1      & \alpha^{3}  & \alpha^{6}      & \alpha^{9}     & \dots & \alpha^{3(n-1)}\\
\vdots& \vdots            & \vdots               &     \vdots          &   \ddots    & \vdots\\
1      & \alpha^{n-1} & \alpha^{2(n-1)} & \alpha^{3(n-1)} & \dots & \alpha^{(n-1)(n-1)}\\
\end{array}\right)
\left(\begin{array}{c}a_0\\a_1\\\vdots\\a_{n-1}\end{array}\right)
= \left(\begin{array}{c}u_0\\u_1\\\vdots\\u_{n-1}\end{array}\right).
$$
The matrix of this system is a square Vandermonde matrix which is known to be invertible. 
So, any $u$ in ${\mathbb F}_q^n$ is 
of the form
$(f(1),f(\alpha),f(\alpha^2),\dots,f(\alpha^{n-1}))$
for some unique $f\in{\mathbb F}_q[x]$ of degree less than $n$.

\begin{example}\label{ex:fufw}
In ${\mathbb F}_7$, taking $\alpha=5$ as primitive element, we have
$$\begin{array}{rcl}
f_0&=&6x^5 + 6x^4 + 6x^3 + 6x^2 + 6x + 6,\\
f_1&=&2x^5 + 3x^4 + x^3 + 5x^2 + 4x + 6,\\
f_2&=&3x^5 + 5x^4 + 6x^3 + 3x^2 + 5x + 6,\\
f_3&=&x^5 +  6x^4 + x^3 + 6x^2 + x + 6,\\
f_4&=&5x^5 + 3x^4 + 6x^3 + 5x^2 + 3x + 6,\\
f_5&=&4x^5 + 5x^4 + x^3 + 3x^2 + 2x + 6.\\
\end{array}$$
Then, for a general vector $u\in{\mathbb F}_7^6$, the coefficients of $f_u$ (in increasing order) can be computed as the product of $u$ by the matrix
$$\left(\begin{array}{cccccc}
6& 6& 6& 6& 6& 6\\
6& 4& 5& 1& 3& 2\\
6& 5& 3& 6& 5& 3\\
6& 1& 6& 1& 6& 1\\
6& 3& 5& 6& 3& 5\\
6& 2& 3& 1& 5& 4\\
\end{array}\right).$$
For instance, the coefficients of the polynomial interpolating $u=(4, 2, 1, 6, 3, 2)$
are $(3, 0, 3, 2, 6, 4)$, and the coefficients of the polynomial interpolating
$w=(0, 2, 5, 6, 0, 6)$ are $(2, 2, 2, 2, 6,0)$.
\end{example}

\paragraph{Code word checking.} From Definition~\ref{def:interpolation}, a vector $u=(u_0,\dots,u_{n-1})$ in ${\mathbb F}_q^n$
is a code word if and only if its interpolation polynomial $f_u$ satisfies $\deg(f_u)<k$.

\begin{example}
  The words  $265034$, $231546$, and $401632$ are code words of $RS_{7,5}(2)$ because, as seen in Example~\ref{e:intpolGf},
  their interpolation polynomials are, respectively, $x+1$, $2x$, and $6x+5$, whose degrees are less than $k=2$.
  The words $421632$ and $025606$ are not code words of $RS_{7,5}(2)$ because, as seen in Example~\ref{ex:fufw}, their interpolation polynomials are, respectively, $6x^5+6x^3+5x^2+2x$ and $6x^4+2x^3+2x^2+2x+2$, whose degrees are larger than $k=2$.
\end{example}

\subsection{Reed--Solomon codes and polynomial evaluation.}

Consider now the set ${\mathbb F}_q[x]^{<n}$ of all polynomials with coefficients in ${\mathbb F}_q$ and degree strictly less than $n$.
A general element $u\in{\mathbb F}_q[x]^{<n}$ is of the form $u=u_0+u_1x+\dots+u_{n-1}x^{n-1}$ with $u_i\in{\mathbb F_q}$. Observe that evaluating $u$ at $\alpha^{i}$ gives $u(\alpha^{i})=u_0+u_1\alpha^{i}+\dots+u_{n-1}\alpha^{i(n-1)},$ which, if $i\leq n-k$, is exactly the result of the product of the $i$th row of matrix $H$ and vector $(u_0,\dots,u_{n-1})^T$.
The value $u(\alpha^i)$, if $i\leq n-k$, is called the $i$th {\it syndrome} of $u$ with respect to $C$.
Now, the product of matrix $H$ and vector $(u_0,\dots,u_{n-1})^T$ is exactly the vector $(u(\alpha),u(\alpha^2),\dots,u(\alpha^{n-k})),$ which is called the {\it syndrome vector} of $u$ with respect to $C$. On the other hand, by definition of parity-check matrix, $(u_0,\dots,u_{n-1})$ is a code word if and only if the product of matrix $H$ and $(u_0,\dots,u_{n-1})^T$ is zero.

\begin{definition}\label{def:evaluation}
The Reed--Solomon code over ${\mathbb F}_q$ and of dimension $k$, $RS_{q,\alpha}(k)$, is the set of vectors $u=(u_0,\dots,u_{n-1})$ in ${\mathbb F}_q^n$ such that the polynomial $u_0+u_1x+\dots+u_{n-1}x^{n-1}$ vanishes at $\alpha^j$ for all $j$ with $1\leq j\leq n-k$.
\end{definition}

\paragraph{Code word checking.}
Now, given a vector $u=(u_0,\dots,u_{n-1})$ in ${\mathbb F}_q^n$, 
$u$ is a code word if and only if $u(\alpha^i)=0$ for all $i$ with $1\leq i\leq n-k$.

  \begin{example}
    \label{e:toyu}
    Suppose we want to check whether the word $342650$ belongs to $RS_{7,5}(2)$.
    We consider the polynomial $u(x)=3+4x+2x^2+6x^3+5x^4$
    and evaluate it at $5$, $5^2$, $5^3$ and $5^4$. We obtain
    $$\begin{array}{rcl}
      u(5^1)=u(5)&=3+6+1+1+3&=0,\\
      u(5^2)=u(4)&=3+2+4+6+6&=0,\\
      u(5^3)=u(6)&=3+3+2+1+5&=0,\\
      u(5^4)=u(2)&=3+1+1+6+3&=0.
    \end{array}$$
Since $u(5)=u(5^2)=u(5^3)=u(5^4)=0$, the word $342650$ belongs to $RS_{7,5}(2)$. 
\end{example}

\subsection{Connection of the coefficients of an interpolation polynomial and its evaluation at all points.}

Next we will see that the coefficients of an interpolation polynomial over a finite field are intimately related to the values obtained when evaluating the polynomial at all the nonzero elements of the finite field.

\begin{lemma}
\label{l:fi}
Suppose that $\alpha$ is a primitive element of a finite field of $q$ elements and let $n=q-1$.
The polynomials $f_i=\prod_{\substack{j=0\\j\neq i}}^{n-1}\frac{x-\alpha^j}{\alpha^i-\alpha^j}$ satisfy
$f_i=-(\alpha^ix^{n-1}+\alpha^{2i}x^{n-2}+\alpha^{3i}x^{n-3}+\dots+\alpha^{(n-1)i}x+\alpha^{ni}).$
\end{lemma}

\begin{proof}
  Suppose $\beta\in{\mathbb F}_q\setminus\{0\}$. From the equality $(x-\beta)(x^{n-1}+\beta x^{n-2}+\beta^2 x^{n-3}+\dots+\beta^{n-2}x+\beta^{n-1})=x^n-1,$ it follows that $x^{n-1}+\beta x^{n-2}+\beta^2 x^{n-3}+\dots+\beta^{n-2}x+\beta^{n-1}=\frac{x^n-1}{x-\beta}.$ This, together with the fact  $x^n-1=\prod_{\gamma\in{\mathbb F}_q\setminus\{0\}}(x-\gamma)$, implies that  $x^{n-1}+\beta x^{n-2}+\beta^2 x^{n-3}+\dots+\beta^{n-2}x+\beta^{n-1}$ vanishes at all the elements of ${\mathbb F}_q\setminus\{0\}$ except at $\beta$, where it evaluates to $\beta^{n-1}+\beta \beta^{n-2}+\beta^2 \beta^{n-3}+\dots+\beta^{n-2}\beta+\beta^{n-1}=n\beta^{n-1}=\frac{-1}{\beta}$. Hence, $-\beta(x^{n-1}+\beta x^{n-2}+\beta^2 x^{n-3}+\dots+\beta^{n-2}x+\beta^{n-1})$ vanishes at all the elements of ${\mathbb F}_q\setminus\{0\}$ except at $\beta$, where it evaluates to $1$. Finally, $-\beta(x^{n-1}+\beta x^{n-2}+\beta^2 x^{n-3}+\dots+\beta^{n-2}x+\beta^{n-1})=-(\beta x^{n-1}+\beta^{2}x^{n-2}+\beta^{3}x^{n-3}+\beta^{4}x^{n-4}+\dots+\beta^{(n-1)}x+\beta^{n})$.

If we take $\beta=\alpha^i$ then $f_i$ and the expression
have degree $q-2$ and take the same values at $q-1$ points, hence are equal.
\end{proof}

The main result relating the last two definitions of Reed--Solomon codes is the following lemma.

\begin{lemma}
\label{l:inverse}
The inverse of the map $$\begin{array}{ccc}{\mathbb F}^n &\rightarrow& {\mathbb F}^n\\
(v_0,\dots,v_{n-1})&\mapsto &(v(\alpha^0),v(\alpha),v(\alpha^2),\dots,v(\alpha^{n-1}))
\end{array}
$$
is
$$\begin{array}{ccc}
(-u(\alpha^n),-u(\alpha^{n-1}),-u(\alpha^{n-2}),\dots,-u(\alpha)) & 
\mathrel{\reflectbox{\ensuremath{\mapsto}}}
& (u_0,\dots,u_{n-1}),\end{array}$$
where $v(\beta)$ is the evaluation of $v_0+v_1x+\dots+v_{n-1}x^{n-1}$ at $\beta$ and $u(\beta)$ is the evaluation of $u_0+u_1x+\dots+u_{n-1}x^{n-1}$ at $\beta$.
\end{lemma}

\begin{proof}
The inverse map is giving the coefficients of the interpolation polynomial $f_u$. By Lemma~\ref{l:fi} we have that $f_u=\sum_{i=0}^{n-1}u_i f_i$, where
$f_i=-(\alpha^ix^{n-1}+\alpha^{2i}x^{n-2}+\alpha^{3i}x^{n-3}+\dots+\alpha^{(n-1)i}x+\alpha^{ni}).$
Now,
\begin{eqnarray*}f_u&=&x^{n-1}
\left(\begin{array}{cccc} u_0& u_1& \cdots & u_{n-1}\end{array}\right)
\left(\begin{array}{c} -1\\ -\alpha\\ \vdots \\ -\alpha^{n-1}\end{array}\right)\\&&+\ \
x^{n-2}
\left(\begin{array}{cccc} u_0& u_1& \cdots & u_{n-1}\end{array}\right)
\left(\begin{array}{c} -1\\ -\alpha^2\\ \vdots\\ -\alpha^{2(n-1)}\end{array}\right)
\\&&+\ \ \cdots
\\&&+\ \ 
x
\left(\begin{array}{cccc} u_0& u_1& \cdots & u_{n-1}\end{array}\right)
\left(\begin{array}{c} -1\\ -\alpha^{n-1}\\ \vdots\\ -\alpha^{(n-1)(n-1)}\end{array}\right)\\&&+
\left(\begin{array}{cccc} u_0& u_1& \cdots & u_{n-1}\end{array}\right)
\left(\begin{array}{c} -1\\ -\alpha^{n}\\ \vdots\\ -\alpha^{n(n-1)}\end{array}\right).
\end{eqnarray*}
So, $f_u=-u(\alpha)x^{n-1}-u(\alpha^2)x^{n-2}-\dots-u(\alpha^n)=\sum_{i=0}^{n-1}(-u(\alpha^{n-i}))x^i$.
\end{proof}

  \begin{example}
    Consider the word $(u_0,u_1,u_2,u_3,u_4,u_5)=(5,4,0,1,2,0)$ and the related polynomial $u=5+4x+x^3+2x^4$.
Its evaluation at the powers of $5$ is
    $$\begin{array}{rcl}
      u(5^1)=u(5)&=5+6+6+4&=0,\\
      u(5^2)=u(4)&=5+2+1+1&=2,\\
      u(5^3)=u(6)&=5+3+6+2&=2,\\
      u(5^4)=u(2)&=5+1+1+4&=4,\\
      u(5^5)=u(3)&=5+5+6+1&=3,\\
      u(5^6)=u(1)&=5+4+1+2&=5.
    \end{array}$$
What Lemma~\ref{l:inverse} says is that the polynomial $v(x)=-5-3x-4x^2-2x^3-2x^4=2+4x+3x^2+5x^3+5x^4$ satisfies that $u=(v(1),v(5),v(5^2),v(5^3),v(5^4),v(5^5)).$ Indeed,   $$\begin{array}{rcl}
  v(5^0)=v(1)&=2+4+3+5+5&=5,\\
      v(5^1)=v(5)&=2+6+5+2+3&=4,\\
      v(5^2)=v(4)&=2+2+6+5+6&=0,\\
      v(5^3)=v(6)&=2+3+3+2+5&=1,\\
      v(5^4)=v(2)&=2+1+5+5+3&=2,\\
      v(5^5)=v(3)&=2+5+6+2+6&=0,\\
    \end{array}$$
and it follows that $(v(1),v(5),v(5^2),v(5^3),v(5^4),v(5^5))=(5,4,0,1,2,0)=(u_0,u_1,u_2,u_3,u_4,u_5)$.
  \end{example}

\section{New decoding approach.}
\label{sec:decap}

We approach decoding from the point of view of Definition~\ref{def:interpolation}. However, we use Definition~\ref{def:evaluation} for the proofs.

Let ${\mathbb F}_q[x]^{<d}$ be the set of polynomials with coefficients in ${\mathbb F}_q$ and degree strictly less than $d$, and let 
${\mathbb F}_q[x]^{<d}_{\geq d'}$ be the set of polynomials with coefficients in ${\mathbb F}_q$ and with only terms of degrees at least $d'$ and at most $d-1$.

Suppose we receive $u\in{\mathbb F}_q^n$. Let $f_u$ be the interpolation polynomial of $u$. Decoding $u$ is the same as finding $c\in RS_{q,\alpha}(k)$ such that $u$ and $c$ are at minimum Hamming distance.
Since words $c\in RS_{q,\alpha}(k)$ are the evaluation of polynomials of degrees smaller than $k$ at the nonzero elements of ${\mathbb F}_q$,
decoding $u$ is equivalent to finding $g_c\in{\mathbb F}_q[x]^{<k}$ such that $f_u-g_c$ has maximum number of nonzero roots. In fact, $g_c$ is then the interpolation polynomial of $c$.

The monomials of $f_u$ can be split into those that have degree less than $k$ and those having degree at least $k$. Let $h_u,g_u$ be the unique polynomials with $h_u\in{\mathbb F}_q[x]^{<n}_{\geq k}$, $g_u\in{\mathbb F}_q[x]^{<k}$ such that $f_u=h_u+g_u$.
Once $f_u$ is fixed, and so is $h_u$, consider, from all the polynomials in ${\mathbb F}_q[x]^{<k}$, a polynomial $g_{h_u}$ that maximizes the number of nonzero roots of $h_u+g_{h_u}$. That is,
the number of nonzero roots of $h_u+g_{h_u}$ is larger than or equal to the number of nonzero roots of $h_u+g'$ for any $g'\in{\mathbb F}_q[x]^{<k}$.
Then,
\begin{equation}
\label{eq:gs}
g_c=g_u-g_{h_u}.
\end{equation}
Notice that if $e$ is the minimum weight word such that $u-e\in RS_{q,\alpha}(k)$, then $e=u-c$ and its interpolation polynomial is $f_e=f_u-g_c=h_u+g_{h_u}$.

Consider the set
$$\Lambda=\{\lambda\in{\mathbb F}_q[x] \mbox{\,such\,that\,} \lambda(h_u+g) \mbox{\,vanishes\,at\,all\,}{\mathbb F}_q\setminus\{0\}, \mbox{ for some }g\in{\mathbb F}_q[x]^{<k}\}.$$
Because of the fact that $x^n-1=\prod_{\gamma\in{\mathbb F}_q\setminus\{0\}}(x-\gamma)$, an equivalent definition is
$$\Lambda=\{\lambda\in{\mathbb F}_q[x] \mbox{ such that } (x^n-1) \mbox{ divides } \lambda(h_u+g) \mbox{ for some }g\in{\mathbb F}_q[x]^{<k}\}.$$
Notice that $\Lambda$ is not empty because $x^n-1$ belongs to $\Lambda$.
\begin{theorem}
\label{t:lambda}
Let $\lambda_u$ be a monic polynomial with minimum degree among the polynomials in $\Lambda$.  For a polynomial $g\in{\mathbb F}_q[x]^{<k}$, if $(x^n-1)$ divides $\lambda_u(h_u+g)$, then the number of nonzero roots of $h_u+g$ is greater than or equal to the number of nonzero roots of $h_u+g'$ for any $g'\in{\mathbb F}_q[x]^{<k}$.
\end{theorem}

\begin{proof}
For a fixed $g\in{\mathbb F}_q[x]^{<k}$, the set of polynomials 
$$\Lambda_g=\{\lambda\in{\mathbb F}_q[x] \mbox{ such that } (x^n-1) \mbox{ divides } \lambda(h_u+g)\}$$
is, since $x^n-1=\prod_{\gamma\in{\mathbb F}_q\setminus\{0\}}(x-\gamma)$,
the set of polynomials that are multiples of 
\begin{equation}\label{eq:prodlambda}\prod_{\substack{\gamma\in{\mathbb F}_q\setminus\{0\}\\(h_u+g)(\gamma)\neq 0}}(x-\gamma).
\end{equation}
The monic polynomial with minimum degree among $\Lambda_g$ is then the polynomial \eqref{eq:prodlambda} itself.
Now, $\Lambda=\cup_{g\in {\mathbb F}_q[x]^{<k}}\Lambda_g$. So a monic polynomial with minimum degree among $\Lambda$ must be one of the polynomials as in \eqref{eq:prodlambda} for some $g\in{\mathbb F}_q[x]^{<k}$. The minimality of the degree of $\lambda_u$ implies the maximality of the number of nonzero roots of $h_u+g$, where $g$ is such that $\lambda_u\in \Lambda_g$.
\end{proof}

Let $\lambda_u$ be as in Theorem~\ref{t:lambda} and suppose $\mu\in{\mathbb F}_q[x]$ is such that $\lambda_u(h_u+g)=\mu(x^n-1)$ for some $g\in{\mathbb F}_q[x]^{<k}$.
Suppose that $\deg(\lambda_u)=t$ and $\deg(h_u)=d_u$. In particular, $t\leq n$ and $d_u\leq n-1$. Now, $\deg(\mu)=t+d_u-n$.

Let the coefficients of $\lambda_u(h_u+g)$ be $\xi_0,\dots,\xi_{d_u+t}$.
If $h_u=a_{d_u}x^{d_u}+a_{d_u-1}x^{d_u-1}
+\dots+a_kx^k$
and $\lambda_u=x^t+l_{t-1}x^{t-1}+\dots+l_1x+l_0,$
then, letting $a_j=0$ for all $j>d_u$ and $\xi_j=0$ for all $j>d_u+t$, we have for all $i\geq 0$,
$$
\begin{array}{ccccccccccc}
\xi_{k+t+i}&=&a_{k+t+i}l_0&+&a_{k+t+i-1}l_{1}&+&\cdots&+&a_{k+i+1}l_{t-1}&+&a_{k+i},\\
\end{array}
$$
which, by Lemma~\ref{l:inverse}, is equivalent to
\begin{equation}\label{eq:linu}
\xi_{k+t+i}=-u(\alpha^{n-k-t-i})l_0-u(\alpha^{n-k-t-i+1})l_{1}-\cdots
-u(\alpha^{n-k-i}).
\end{equation}
Since $\deg(\mu)=t+d_u-n<n$, the coefficients of $\lambda_u(h_u+g)=\mu(x^n-1)=\mu x^n-\mu$ satisfy 
$\xi_i=-\xi_n$, $\xi_1=-\xi_{n+1}$, \dots, $\xi_{t+d_u-n}=-\xi_{t+d_u}$
and
$\xi_{t+d_u-n+1}=\xi_{t+d_u-n+2}=\dots=\xi_{n-1}=0.$

\begin{lemma}
\label{l:sys}
Let $\lambda_u$ be as in Theorem~\ref{t:lambda}.
The nonleading coefficients of $\lambda_u$ give a solution to the linear system
\setlength{\arraycolsep}{2pt}
\begin{equation}\label{eq:linsysu}\left(
\begin{array}{cccc}
u(\alpha) & u(\alpha^{2}) & \dots & u(\alpha^t)\\
u(\alpha^2) & u(\alpha^3) & \dots & u(\alpha^{t+1})\\
\vdots & \vdots & \ddots & \vdots \\
u(\alpha^{n-k-t}) & u(\alpha^{n-k-t+1}) & \dots  & u(\alpha^{n-k-1}) \\
\end{array}
\right)
\left(
\begin{array}{c}
l_0\\l_1\\\vdots \\l_{t-1}
\end{array}
\right)=
\left(
\begin{array}{c}
-u(\alpha^{t+1})\\
-u(\alpha^{t+2})\\
\vdots\\
-u(\alpha^{n-k})
\end{array}
\right).\end{equation}
\end{lemma}

\begin{proof}
The lemma is a consequence of equation \eqref{eq:linu}
and the fact that $\xi_{k+t},\dots,\xi_{n-1}=0,$ since $k+t\geq d_u+t-n+1$.
\end{proof}

\begin{lemma}
\label{l:linearsystems}
Let $t$ be the weight of a minimum weight vector $e\in{\mathbb F}_q^n$ such that $u-e\in RS_{q,\alpha}(k)$ and consider the linear system 
\setlength{\arraycolsep}{3pt}
\begin{equation*}\left(
\begin{array}{cccc}
u(\alpha) & u(\alpha^{2}) & \dots & u(\alpha^{t'})\\
u(\alpha^2) & u(\alpha^3) & \dots & u(\alpha^{t'+1})\\
\vdots & \vdots & \ddots & \vdots \\
u(\alpha^{n-k-t'}) & u(\alpha^{n-k-t'+1}) & \dots  & u(\alpha^{n-k-1}) \\
\end{array}
\right)
\left(
\begin{array}{c}
l_0\\l_1\\\vdots \\l_{t'-1}
\end{array}
\right)=
\left(
\begin{array}{c}
-u(\alpha^{t'+1})\\
-u(\alpha^{t'+2})\\
\vdots\\
-u(\alpha^{n-k})
\end{array}
\right).\end{equation*}
\setlength{\arraycolsep}{5pt}

\begin{enumerate}
\item If $t\leq \frac{n-k}{2}$ and $t'=t$, then the linear system has a unique solution,
which can be found as a solution to the square system
\begin{equation*}\left(
\begin{array}{cccc}
u(\alpha) & u(\alpha^{2}) & \dots & u(\alpha^{t})\\
u(\alpha^2) & u(\alpha^3) & \dots & u(\alpha^{t+1})\\
\vdots & \vdots & \ddots & \vdots \\
u(\alpha^{t}) & u(\alpha^{t+1}) & \dots  & u(\alpha^{2t-1}) \\
\end{array}
\right)
\left(
\begin{array}{c}
l_0\\l_1\\\vdots \\l_{t-1}
\end{array}
\right)=
\left(
\begin{array}{c}
-u(\alpha^{t+1})\\
-u(\alpha^{t+2})\\
\vdots\\
-u(\alpha^{2t})
\end{array}
\right).\end{equation*}

\item If $t\leq \frac{n-k}{2}$ and $t'=t$, then the unique solution
to the previous system satisfies $l_0\neq 0$.
\item If $t\leq \frac{n-k}{2}$ and $t'<t$, then the system has no solution.
\end{enumerate}
\end{lemma}

\begin{proof}
\begin{enumerate}
\item
The existence of a solution is a consequence of Lemma~\ref{l:sys}.
For the uniqueness, we will see that the square submatrix 
$$\left(
\begin{array}{cccc}
u(\alpha) & u(\alpha^{2}) & \dots & u(\alpha^{t})\\
u(\alpha^2) & u(\alpha^3) & \dots & u(\alpha^{t+1})\\
\vdots & \vdots & \ddots & \vdots \\
u(\alpha^{t}) & u(\alpha^{t+1}) & \dots  & u(\alpha^{2t-1}) \\
\end{array}
\right)$$
has nonzero determinant.
As a consequence of Definition~\ref{def:evaluation} of $RS_{q,\alpha}(k)$ and the fact that $2t-1\leq n-k$,
\begin{equation*}
\left(\begin{array}{ccc}
u(\alpha) & \dots  & u(\alpha^{t}) \\
\vdots &  \ddots & \vdots \\
u(\alpha^{t}) & \dots & u(\alpha^{2t-1})
\end{array}
\right)
=\left(
\begin{array}{ccc}
e(\alpha) & \dots  & e(\alpha^{t}) \\
\vdots &  \ddots & \vdots \\
e(\alpha^{t}) & \dots & e(\alpha^{2t-1})
\end{array}
\right).\end{equation*}

Suppose that the nonzero positions of $e$ are $i_1,\dots,i_t$, with $0\leq i_1<i_2<\dots<i_t\leq n-1$. Then it is easy to check that, letting 
$$W=\left(\begin{array}{cccc}
1    & \dots         & \dots &          1\\
\alpha^{i_1}      & \alpha^{i_2}      & \dots & \alpha^{i_t}\\
\alpha^{2i_1}    & \alpha^{2i_2}      & \dots & \alpha^{2i_t}\\
\vdots      & \vdots        & \ddots& \vdots\\
\alpha^{(t-1)i_1} & \alpha^{(t-1)i_2} & \dots & \alpha^{(t-1)i_t}
\end{array}
\right),$$ we have

\begin{equation*}\left(
\begin{array}{ccc}
e(\alpha) & \dots  & e(\alpha^{t}) \\
\vdots &  \ddots & \vdots \\
e(\alpha^{t}) & \dots & e(\alpha^{2t-1})
\end{array}
\right)
=W
\left(\begin{array}{cccc}
\alpha^{i_1}e_{i_1}      &0      & \dots &          0\\
0      & \alpha^{i_2}e_{i_2}   & \ddots &\vdots\\
\vdots      & \ddots        & \ddots& 0\\
0 & \dots  & 0  &\alpha^{i_t}e_{i_t}
\end{array}
\right)
W^T,\end{equation*}
which clearly has nonzero determinant because $W$ is a Vandermonde matrix.

\item Suppose that
\begin{equation*}\left(
\begin{array}{cccc}
u(\alpha) & u(\alpha^{2}) & \dots & u(\alpha^{t})\\
u(\alpha^2) & u(\alpha^3) & \dots & u(\alpha^{t+1})\\
\vdots & \vdots & \ddots & \vdots \\
u(\alpha^{t}) & u(\alpha^{t+1}) & \dots  & u(\alpha^{2t-1}) \\
\end{array}
\right)
\left(
\begin{array}{c}
0\\l_1\\\vdots \\l_{t-1}
\end{array}
\right)=
\left(
\begin{array}{c}
-u(\alpha^{t+1})\\
-u(\alpha^{t+2})\\
\vdots\\
-u(\alpha^{2t})
\end{array}
\right).\end{equation*}
Then, rearranging the columns, and considering Definition~\ref{def:evaluation} of $RS_{q,\alpha}(k)$ together with the fact that $2t\leq n-k$, we obtain
\begin{equation}\label{eq:001}
\left(
\begin{array}{cccc}
e(\alpha^{2}) & \dots & e(\alpha^{t}) & e(\alpha^{t+1})\\
e(\alpha^3) & \dots & e(\alpha^{t+1}) & e(\alpha^{t+2})\\
\vdots & \ddots & \ddots & \vdots\\
e(\alpha^{t+1}) & \dots  & e(\alpha^{2t-1}) & e(\alpha^{2t})\\
\end{array}
\right)
\left(
\begin{array}{c}
l_1\\\vdots \\l_{t-1}\\ 1
\end{array}
\right)=
\left(
\begin{array}{c}
0\\
0\\
\vdots\\
0
\end{array}
\right),\end{equation}
but 
\begin{equation*}\left(
\begin{array}{ccc}
e(\alpha^{2}) & \dots  & e(\alpha^{t+1})\\
e(\alpha^3) & \dots  & e(\alpha^{t+2})\\
\vdots & \ddots & \vdots\\
e(\alpha^{t+1}) & \dots  & e(\alpha^{2t})\\
\end{array}
\right)=
W
\left(\begin{array}{ccc}
\alpha^{2i_1}e_{i_1}      &\dots      &   0\\
 \vdots   &\ddots     & \vdots\\
0 &\dots &\alpha^{2i_t}e_{i_t}
\end{array}
\right)
W^T,
\end{equation*}
which, again, has nonzero determinant. This contradicts \eqref{eq:001}.
\item
Suppose $t'<t$ and $t-t'=\delta$.
If 
\begin{equation*}\left(
\begin{array}{ccc}
u(\alpha) &  \dots & u(\alpha^{t'})\\
u(\alpha^2)  & \dots & u(\alpha^{t'+1})\\
\vdots &  \ddots & \vdots \\
u(\alpha^{n-k-t'}) & \dots  & u(\alpha^{n-k-1}) \\
\end{array}
\right)
\left(
\begin{array}{c}
\gamma_1\\
\gamma_2\\
\vdots\\
\gamma_{t'}
\end{array}
\right)=
\left(
\begin{array}{c}
-u(\alpha^{t'+1})\\
-u(\alpha^{t'+2})\\
\vdots\\
-u(\alpha^{n-k})
\end{array}
\right),\end{equation*}
then, supressing the first $\delta$ rows we obtain
\setlength{\arraycolsep}{2pt}
\begin{equation*}
\left(
\begin{array}{cccc}
u(\alpha^{\delta+1}) & u(\alpha^{\delta+2}) & \dots & u(\alpha^{\delta+t'}) \\
u(\alpha^{\delta+2}) & u(\alpha^{\delta+3}) & \dots & u(\alpha^{\delta+t'+1})\\
\vdots & \vdots &  \ddots & \vdots\\
u(\alpha^{n-k-t'}) & u(\alpha^{n-k-t'+1}) & \dots  & u(\alpha^{n-k-1}) \\
\end{array}
\right)
\left(
\begin{array}{c}
\gamma_1\\
\gamma_2\\
\vdots\\
\gamma_{t'}
\end{array}
\right)=\left(
\begin{array}{c}
-u(\alpha^{t+1})\\
-u(\alpha^{t+2})\\
\vdots\\
-u(\alpha^{n-k})
\end{array}
\right)\end{equation*}
and adding $\delta$ columns at the beginning,
\setlength{\arraycolsep}{1.5pt}
\begin{equation*}
\left(
\begin{array}{ccccc}
u(\alpha) & \dots & u(\alpha^{\delta+1}) & \dots & u(\alpha^{\delta+t'}) \\
u(\alpha^2) & \dots & u(\alpha^{\delta+2}) & \dots & u(\alpha^{\delta+t'+1})\\
\vdots & \vdots & \vdots &  \vdots & \vdots\\
u(\alpha^{n-k-t}) & \dots &u(\alpha^{n-k-t'}) & \dots  & u(\alpha^{n-k-1}) \\
\end{array}
\right)
\left(
\begin{array}{c}
0\\\vdots \\0\\
\gamma_1\\
\gamma_2\\
\vdots\\
\gamma_{t'}
\end{array}
\right)=\left(
\begin{array}{c}
-u(\alpha^{t+1})\\
-u(\alpha^{t+2})\\
\vdots\\
-u(\alpha^{n-k})
\end{array}
\right).\end{equation*}
\end{enumerate}
This contradicts the two previous points.
\end{proof}
We obtain the following decoding algorithm for a $RS_{q,\alpha}(k)$ code, where 
\mbox{$n=q-1$.}

\bigskip

{\bf Input:} $u\in{\mathbb F}_q^n$.

\begin{enumerate}
\item Let $t$ be the minimum integer such that
  \setlength{\arraycolsep}{1pt}
  \begin{equation}\label{e:ranks}
      \rm{rank}\left(
\begin{array}{ccc}
u(\alpha) & \dots  & u(\alpha^{t}) \\
u(\alpha^{2}) & \dots  & u(\alpha^{t+1}) \\
\vdots & \ddots & \vdots \\
u(\alpha^{n-k-t}) & \dots & u(\alpha^{n-k-1})
\end{array}
\right)
=
\rm{rank}
\left(
\begin{array}{ccccc}
u(\alpha)  & \dots  & u(\alpha^{t+1}) \\
u(\alpha^{2}) & \dots  & u(\alpha^{t+2}) \\
\vdots & \ddots & \vdots \\
u(\alpha^{n-k-t}) & \dots & u(\alpha^{n-k})
\end{array}
\right).
  \end{equation}
\setlength{\arraycolsep}{5pt}
For $t=0$, the first matrix is the null matrix. In this case we consider $\rank()=0$.

\item Solve the linear system
\begin{equation*}\left(
\begin{array}{cccc}
u(\alpha) & u(\alpha^{2}) & \dots & u(\alpha^{t})\\
u(\alpha^2) & u(\alpha^3) & \dots & u(\alpha^{t+1})\\
\vdots & \vdots & \ddots & \vdots \\
u(\alpha^{t}) & u(\alpha^{t+1}) & \dots  & u(\alpha^{2t-1}) \\
\end{array}
\right)
\left(
\begin{array}{c}
l_0\\l_1\\\vdots \\l_{t-1}
\end{array}
\right)=
\left(
\begin{array}{c}
-u(\alpha^{t+1})\\
-u(\alpha^{t+2})\\
\vdots\\
-u(\alpha^{2t})
\end{array}
\right)\end{equation*}
for $l_0\ldots,l_{t-1}$ and denote by $\lambda_u$ the polynomial $x^t+l_{t-1}x^{t-1}+\dots+l_1x+l_0$.
\item
Obtain as in Lemma~\ref{l:inverse} the interpolation polynomial $f_u$ of $u$, and let $d_u$ be its degree.
\item 
Let $\zeta_{0},\dots,\zeta_{d_u+t}$ be the coefficients of $\lambda_uf_u$; that is, $\lambda_uf_u=\zeta_0+\zeta_1x+\dots+\zeta_{d_u+t} x^{d_u+t}.$

Let $g_c=f_u-\frac{(x^n-1)(\zeta_n+\zeta_{n+1}x+\dots+\zeta_{d_u+t}x^{d_u+t-n})}{\lambda_u}$.

\item {\bf Output:} $(g_c(1),g_c(\alpha),g_c(\alpha^2),\dots,g_c(\alpha^{n-1}))$.
\end{enumerate}

\begin{theorem}
Suppose we received $u\in{\mathbb F}_q^n$.
Let $t$ be the weight of a minimum weight vector $e\in{\mathbb F}_q^n$ such that $u-e\in RS_{q,\alpha}(k)$. If $t\leq \frac{n-k}{2}$, then the previous algorithm outputs $u-e$.
\end{theorem}

\begin{proof}
By Lemma~\ref{l:linearsystems}, step 1 gives the actual number of errors $t$.
By Lemma~\ref{l:linearsystems} again, the system in step 2 has a unique solution and, by Lemma~\ref{l:sys}, the polyomial one obtains is exactly the polynomial $\lambda_u$ in Theorem~\ref{t:lambda}.
After step 3 we get the interpolation polynomial $f_u$.
Let $h_u,g_u$ be the unique polynomials with $h_u\in{\mathbb F}_q[x]^{<n}_{\geq k}$, $g_u\in{\mathbb F}_q[x]^{<k}$ such that $f_u=h_u+g_u$.
By Theorem~\ref{t:lambda}, $x^n-1$ divides $\lambda_u(h_u+g_{h_u})$ for some $g_{h_u}$ maximizing the number of nonzero roots of $h_u+{\mathbb F}_q[x]^{<k}$. In particular, there exists $\mu$ such that 
\begin{equation}
\label{eq:mu}
(x^n-1)\mu=\lambda_u(h_u+g_{h_u})
\end{equation} 
and $\mu$ must have degree less than $n$. Hence, the degrees of the monomials in $x^n\mu$ and those in $-\mu$ do not overlap. On the other hand, the monomials of $\lambda_u g_{h_u}$ have degree less than $t+k\leq \frac{n-k}{2}+k=\frac{n+k}{2}\leq n$. So, the monomials in $x^n\mu$ of degrees at least $n$ coincide with the monomials in $\lambda_u h_u$ of degrees at least $n$. That is, $x^n\mu=\zeta_nx^n+\zeta_{n+1}x^{n+1}+\dots+\zeta_{d_u+t}x^{d_u+t}$, and we deduce that $\mu=\zeta_n+\zeta_{n+1}x+\dots+\zeta_{d_u+t}x^{d_u+t-n}$. Now, from \eqref{eq:mu}, we deduce that 
$g_{h_u}=\frac{(x^n-1)(\zeta_n+\zeta_{n+1}x+\dots+\zeta_{d_u+t}x^{d_u+t-n})}{\lambda_u}-h_u$.
Now, as explained in equation \eqref{eq:gs}, the polynomial $g_c$ interpolating the code word $c\in RS_{q,\alpha}(k)$ at minimum distance of $u$ is
$g_c=g_u-g_{h_u}=f_u-\frac{(x^n-1)(\zeta_n+\zeta_{n+1}x+\dots+\zeta_{d_u+t}x^{d_u+t-n})}{\lambda_u}$.
From here it follows that the output is, indeed, the code word $c\in RS_{q,\alpha}(k)$ at minimum distance of $u$.
\end{proof}

\begin{remark}
  Steps 3 and 4 of the algorithm can be replaced by the equivalent steps in the Peterson--Gorenstein--Zierler algorithm, that is, we can find the error positions by means of the roots of $\lambda_u$ and then obtain the error values by means of the linear system
  $$\left(
\begin{array}{cccc}

\alpha^{i_1} & \alpha^{i_2} & \dots & \alpha^{i_t}\\
\alpha^{2i_1} & \alpha^{2i_2} & \dots & \alpha^{2i_t}\\
\vdots & \vdots & \ddots & \vdots \\
\alpha^{ti_1} & \alpha^{ti_2} & \dots & \alpha^{ti_t}\\
\end{array}
\right)
\left(
\begin{array}{c}
e_{i_1}\\e_{i_2}\\\vdots \\e_{i_t}
\end{array}
\right)=
\left(
\begin{array}{c}
u(\alpha)\\
u(\alpha^2)\\
\vdots\\
u(\alpha^t)
\end{array}
\right).$$
\end{remark}

\begin{example}
  Consider the same code as in Example~\ref{e:Gf}, that is, the code $C=RS_{7,5}(2)$. Suppose that after transmission of three code words we receive $u=421632$, $v=342650$, $w=025606$.

  Denote by the same symbol $u$ the vector $421632$ and the polynomial $4+2x+x^2+6x^3+3x^4+2x^5$. The syndromes of $u$ are
    $$\left(\begin{array}{c}
      u(a)\\u(a^2)\\u(a^3)\\u(a^4)
    \end{array}\right)=
    \left(
  \begin{array}{cccccc}
    1 & 5 & 4 & 6 & 2 & 3 \\
    1 & 4 & 2 & 1 & 4 & 2 \\
    1 & 6 & 1 & 6 & 1 & 6 \\
    1 & 2 & 4 & 1 & 2 & 4 \\
  \end{array}
  \right)
\left(\begin{array}{c}
      4\\2\\1\\6\\3\\2
\end{array}\right)=
\left(\begin{array}{c}
      3\\1\\5\\4
    \end{array}\right).
  $$
Since the syndromes are nonzero we deduce that there is at least one error.
We have $\rank\left(\begin{array}{c}\\ \\ \\\end{array}\right)\neq \rank\left(\begin{array}{c}3\\1\\5\\4\end{array}\right)$, but $\rank\left(\begin{array}{c}3\\1\\5\end{array}\right)= \rank\left(\begin{array}{cc}3&1\\1&5\\5&4\end{array}\right)$. So $t=1$ and there is only one error.
 We solve the system 
 $3 l_0=-1$, whose solution is $l_0=2$.
 We deduce that the error locator polynomial is
 $\lambda=x+2$.
 We compute $f_u$ as in Example~\ref{ex:fufw}, obtaining $f_u=4x^5 + 6x^4 + 2x^3 + 3x^2 + 3$. Now,
 since $f_u\cdot \lambda=4x^6 + 6x^2 + 3x + 6$, we deduce that 
$\zeta_n+\dots+\zeta_{d_u+t}x^{d_u+t-n}=4$
 and $g_c=6x+5$, so that the corrected word is
$(401632)$.
 We could also have 
 found the single root of $\lambda$, which is
 $5=5^1$, and then deduce that there is an error at the second position (first position, if we start counting by $0$).
 Then, to find the error value we could have solved the system
 $5^1 e_1=u(a)=3$, whose solution is
 $e_1=2$.  The corrected word is then $u-(020000)=(421632)-(020000)=(401632)$.

  Denote by the same symbol $v$ the vector
  $342650$
  and the polynomial
  $3+4x+2x^2+6x^3+5x^4$.
  The syndromes of $v$ are
    $$\left(\begin{array}{c}
      v(a)\\v(a^2)\\v(a^3)\\v(a^4)
    \end{array}\right)=
    \left(
  \begin{array}{cccccc}
    1 & 5 & 4 & 6 & 2 & 3 \\
    1 & 4 & 2 & 1 & 4 & 2 \\
    1 & 6 & 1 & 6 & 1 & 6 \\
    1 & 2 & 4 & 1 & 2 & 4 \\
  \end{array}
  \right)
\left(\begin{array}{c}
      3\\4\\2\\6\\5\\0
\end{array}\right)=
\left(\begin{array}{c}
      0\\0\\0\\0
    \end{array}\right).
  $$
Since the syndromes are all zero we deduce that there is no error.

 Denote by the same symbol $w$ the vector
$025606$ 
  and the polynomial
  $2x+5x^2+6x^3+6x^5$.
  The syndromes of $w$ are
      $$\left(\begin{array}{c}
      w(a)\\w(a^2)\\w(a^3)\\w(a^4)
    \end{array}\right)=
    \left(\begin{array}{cccccc}
    1 & 5 & 4 & 6 & 2 & 3 \\
    1 & 4 & 2 & 1 & 4 & 2 \\
    1 & 6 & 1 & 6 & 1 & 6 \\
    1 & 2 & 4 & 1 & 2 & 4 \\
    \end{array}\right)
    \left(\begin{array}{c}
      0\\2\\5\\6\\0\\6
    \end{array}\right)=
    \left(\begin{array}{c}
      0\\1\\5\\5
    \end{array}\right).$$
Since the syndromes are nonzero we deduce that there is at least one error.
We have $\rank\left(\begin{array}{c} \\ \\ \\ \end{array}\right)\neq \rank\left(\begin{array}{c}0\\1\\5\\5\end{array}\right)$, $\rank\left(\begin{array}{c}0\\1\\5\end{array}\right)\neq \rank\left(\begin{array}{cc}0&1\\1&5\\5&5\end{array}\right)$, while
$\rank\left(\begin{array}{cc}0&1\\1&5\end{array}\right)= \rank\left(\begin{array}{ccc}0&1&5\\1&5&5\end{array}\right)=2$. So, $t=2$.
We solve the system
$$\left(\begin{array}{cc}0 & 1\\ 1 & 5\end{array}\right)\left(\begin{array}{c}l_0\\l_1\end{array}\right)=\left(\begin{array}{c}-5\\-5\end{array}\right)=\left(\begin{array}{c}2\\2\end{array}\right),$$
        whose solution is $l_0=6$, $l_1=2$.
We deduce that the error locator polynomial is         
$\lambda=x^2+2x+6.$
We compute $f_w$ as in Example~\ref{ex:fufw}, obtaining $f_w=6x^4 + 2x^3 + 2x^2 + 2x + 2$. Now, since $f_w\cdot\lambda=6x^6 + 4x^3 + 4x^2 + 2x + 5$, we deduce that $\zeta_n+\dots+\zeta_{d_u+t}x^{d_u+t-n}=6$ and $g_c=4x+3$, so that the corrected word is
 $(025641)$.
We could also have found the roots of $\lambda$, which are $2=5^4$ and $3=5^5$ and then deduce that
the error positions are the fifth and sixth positions (fourth and fifth positions, if we start counting by $0$).
 Then, to find the error value we could have solved the system
 $$\left(\begin{array}{cc}5^4 & 5^5\\ 5^2 & 5^4 \end{array}\right)\left(\begin{array}{c}e_4\\e_5\end{array}\right)=\left(\begin{array}{cc}2 & 3\\ 4 & 2 \end{array}\right)\left(\begin{array}{c}e_4\\e_5\end{array}\right)=\left(\begin{array}{c}w(a)\\w(a^2)\end{array}\right)=\left(\begin{array}{c}0\\1\end{array}\right),$$
whose solution is $e_4=3$, $e_5=5$.
The corrected word is  then $w-(000035)=(025606)-(000035)=(025641)$.
  
\end{example}

\section{A glimpse of the Peterson--Gorenstein--Zierler algorithm.}

Peterson \cite{Peterson} and Gorenstein and Zierler \cite{GorensteinZierler} proposed a decoding algorithm which is very similar to the one we just presented.
It is based on the following lemma.

For all $h$ with $t\leq h\leq \frac{n-k}{2}$, define
$$A_h=
\left(
\begin{array}{cccc}
u(\alpha) & u(\alpha^2) & \ldots & u(\alpha^h) \\
u(\alpha^2) & u(\alpha^3) & \ldots & u(\alpha^{h+1}) \\
\vdots & \vdots & \ddots & \vdots \\
u(\alpha^h) & u(\alpha^{h+1}) & \ldots & u(\alpha^{2h-1}) \\
\end{array}
\right).$$

\begin{lemma}
If $t<h\leq \frac{n-k}{2}$, then $\det(A_h)=0$, while
$\det(A_t)\not = 0.$
That is, the number of errors (if it is at most $\frac{n-k}{2}$)
is the maximum integer $h\leq \frac{n-k}{2}$ such that
$\det(A_h)\neq 0$.
\end{lemma}

\begin{proof}
Since $2h-1\leq n-k$, 
$$A_h=
\left(
\begin{array}{cccc}
e(\alpha) & e(\alpha^2) & \ldots & e(\alpha^h) \\
e(\alpha^2) & e(\alpha^3) & \ldots & e(\alpha^{h+1}) \\
\vdots & \vdots & \ddots & \vdots \\
e(\alpha^h) & e(\alpha^{h+1}) & \ldots & e(\alpha^{2h-1}) \\
\end{array}
\right).$$
As before, let us denote the error positions as $i_1,\dots,i_t$ and let $M=\{m_1,...,m_h\}\subseteq\{0,...,n-1\}$
be any subset containing all the error positions.
Let $D$ be the diagonal matrix 
$$D=diag(e_{m_1},\dots,e_{m_h}).$$ Clearly, $\left|D\right|\neq0$ if $h=t$
and $\left|D\right|=0$ if $h>t$.
Let
$$W=\left(\begin{array}{cccc}
1    & 1         & \dots &          1\\
\alpha^{m_1}      & \alpha^{m_2}      & \dots & \alpha^{m_h}\\
\alpha^{2m_1}    & \alpha^{2m_2}      & \dots & \alpha^{2m_h}\\
\vdots      & \vdots        & \ddots& \vdots\\
\alpha^{(h-1)m_1} & \alpha^{(h-1)m_2} & \dots & \alpha^{(h-1)m_h}
\end{array}
\right).$$ 
Since $W$ is a Vandermonde matrix and the indices in $M$ are all different,
$\left|W\right|\neq 0$. We have
\begin{equation*}\left(
\begin{array}{ccc}
e(\alpha) & \dots  & e(\alpha^{h}) \\
\vdots &  \ddots & \vdots \\
e(\alpha^{h}) & \dots & e(\alpha^{2h-1})
\end{array}
\right)
=W
\left(\begin{array}{ccc}
\alpha^{m_1}e_{m_1}            & \dots &          0\\
\vdots            & \ddots& \vdots\\
 0  & \dots  &\alpha^{m_h}e_{m_h}
\end{array}
\right)
W^T.
\end{equation*}
Now it is straightforward to check that this product of matrices has zero determinant if and only if $M$ contains no error positions, that is, if $e_{s}=0$ for some $s\in M$.
\end{proof}

The Peterson--Gorenstein--Zierler algorithm is as follows.

\medskip

\noindent{\bf Input:} $u\in{\mathbb F}_q^n$.

\begin{enumerate}
\item Let $t$ be the maximum integer smaller than or equal to $\frac{n-k}{2}$ such that 
$$\left|\begin{array}{cccc}
u(\alpha) & u(\alpha^2) & \ldots & u(\alpha^h) \\
u(\alpha^2) & u(\alpha^3) & \ldots & u(\alpha^{h+1}) \\
\vdots & \vdots & \ddots & \vdots \\
u(\alpha^h) & u(\alpha^{h+1}) & \ldots & u(\alpha^{2h-1}) \\
\end{array}\right|\neq 0.$$

\item Solve the linear system 
\begin{equation*}\left(
\begin{array}{cccc}
u(\alpha) & u(\alpha^{2}) & \dots & u(\alpha^{t})\\
u(\alpha^2) & u(\alpha^3) & \dots & u(\alpha^{t+1})\\
\vdots & \vdots & \ddots & \vdots \\
u(\alpha^{t}) & u(\alpha^{t+1}) & \dots  & u(\alpha^{2t-1}) \\
\end{array}
\right)
\left(
\begin{array}{c}
l_0\\l_1\\\vdots \\l_{t-1}
\end{array}
\right)=
\left(
\begin{array}{c}
-u(\alpha^{t+1})\\
-u(\alpha^{t+2})\\
\vdots\\
-u(\alpha^{2t})
\end{array}
\right)\end{equation*}
for $l_0,\ldots,l_{t-1}$ and denote by $\lambda_u$ the polynomial $x^t+l_{t-1}x^{t-1}+\dots+l_1x+l_0$.
\item Find the error positions by means of the roots of $\lambda_u$.
\item Find the error values by means of the linear system
$$\left(
\begin{array}{cccc}

\alpha^{i_1} & \alpha^{i_2} & \dots & \alpha^{i_t}\\
\alpha^{2i_1} & \alpha^{2i_2} & \dots & \alpha^{2i_t}\\
\vdots & \vdots & \ddots & \vdots \\
\alpha^{ti_1} & \alpha^{ti_2} & \dots & \alpha^{ti_t}\\
\end{array}
\right)
\left(
\begin{array}{c}
e_{i_1}\\e_{i_2}\\\vdots \\e_{i_t}
\end{array}
\right)=
\left(
\begin{array}{c}
u(\alpha)\\
u(\alpha^2)\\
\vdots\\
u(\alpha^t)
\end{array}
\right).$$
\item {\bf Output:} $c=u-e$.
\end{enumerate}

\section{Comparison of our algorithm with the Peterson--Go\-ren\-stein--Zierler algorithm.}
\label{sec:comp}

The main differences between our proposed algorithm and the Peterson--Gorenstein--Zierler algorithm are
\begin{itemize}
\item Computation of the number of errors (step 1 in both algorithms);
\item Computation of the error values (steps 3--4 in both algorithms).
\end{itemize}
Error location (step 2) is done exactly in the same way.

As for the computation of the error values, step 4 in our algorithm needs two polynomial multiplications and one division (all of them of order $t+n$), while steps 3 and 4 in the Peterson--Gornestein--Zierler algorithm involve two linear square systems of $t$ equations. This already makes our algorithm simpler than the Peterson--Gorenstein--Zierler algorithm.

But the main difference is in the computation of the number of errors. In the Peterson--Gorenstein--Zierler algorithm we start computing the determinant of a (large) $\frac{n-k}{2}\times\frac{n-k}{2}$ matrix and continue computing determinants of decreasing order, while in our algorithm we start computing the rank of a (small) $2\times (n-k-1)$ matrix and continue computing ranks of $(h+1)\times (n-k-h)$ matrices with an increasing value of $h$.
In the Peterson--Gorenstein--Zierler algorithm, the smaller the number of errors, the more determinant computations will be needed.
In our algorithm, the smaller the number of errors, the fewer rank computations will be needed. Furthermore, in the Peterson--Gorenstein--Zierler algorithm we start with the most complex determinants and then they get simpler, while in our algorithm we start with the simpler rank computations and, as the number of errors increases we get more complex rank computations.

\section{The optimistic view of best case decoding.}
\label{sec:opt}

There are many scenarios where a high reliability is required but errors rarely occur. In this case, error correcting codes with a high error correction capability are required although the expected number of errors is low.
In terms of error correction, the decoding approach that takes this perspective into consideration is called {\it best case decoding}. See \cite{Berlekamp1996} for a deep analysis. In Berlekamp's clarifying words,
\begin{quote} \dots a best case decoder is philosophically analogous to a small child who continually asks, ``Are we almost there now?'' This question may occur at many places in a long decoding program. But, in a high-reliability application, the odds are quite favorable that any time the question is asked, the answer
is likely to be ``YES''.\end{quote}

Our algorithm is very well suited for best case decoding because, in contrast to the Peterson--Gorenstein--Zierler algorithm, it is fast when the number of errors is small and it is only a bit slower when the number of errors approaches the correction capability.

\section*{Acknowledgments.}
The author would like to thank the anonymous referees for deeply reading the manuscript and for making very useful comments. She would specially like to thank the editor Susan Jane Colley for her careful reading.


\providecommand{\bysame}{\leavevmode\hbox to3em{\hrulefill}\thinspace}
\providecommand{\MR}{\relax\ifhmode\unskip\space\fi MR }
\providecommand{\MRhref}[2]{%
  \href{http://www.ams.org/mathscinet-getitem?mr=#1}{#2}
}
\providecommand{\href}[2]{#2}

\end{document}